\documentclass[12pt,a4paper]{article}
\usepackage{amsmath,amsthm,amssymb,amscd, a4wide}

\usepackage{dsfont}
\usepackage{mathrsfs}%
\usepackage{setspace}
\usepackage{hyperref}

\usepackage{enumerate} 
\usepackage{mathtools} 
\usepackage{latexsym,amsfonts,bbm, mathabx} 
\usepackage{color}
\usepackage{changes}

\newcommand{\xjomega}[1]{x_{#1}^{\omega}}
\newcommand{\DensityMatrixNbeta}{\varrho_N^{\beta,\omega}}
\newcommand{\DensityMatrixNbetaeins}{\varrho_N^{\beta,(1),\omega}}
\newcommand{\AverageParticleDensityNbeta}{\rho_N^{\beta,\omega}}
\newcommand{\EnsembleExpectation}[1]{\langle #1 \rangle_{\DensityMatrixNbeta}}

\newcommand{\ud}{\mathrm{d}}

\DeclareMathOperator{\tr}{Tr}

\DeclareMathOperator{\supp}{supp}

\numberwithin{equation}{section}

\newtheorem{theorem}{Theorem}[section]
\newtheorem{lemma}[theorem]{Lemma}

\newtheorem{remark}[theorem]{Remark}

\theoremstyle{definition}
\newtheorem{defn}[theorem]{Definition}

\numberwithin{equation}{section}

\begin{document}
	
	\allowdisplaybreaks[1]
	
	\thispagestyle{empty}
	
	\vspace*{1cm}
	
	\begin{center}
		
		{\Large \bf  On the effect of repulsive interactions on Bose--Einstein condensation in the Luttinger--Sy model} \\

		\vspace*{2cm}
		
		{\large  Joachim~Kerner \footnote{E-mail address: {\tt joachim.kerner@fernuni-hagen.de}} }%
		
		\vspace*{5mm}
		
		Department of Mathematics and Computer Science\\
		FernUniversität in Hagen\\
		58084 Hagen\\
		Germany\\
		
		\vspace*{2cm}

		{\large  Maximilian Pechmann \footnote{E-mail address: {\tt mpechmann@utk.edu}} }%
		
		\vspace*{5mm}
		
		Department of Mathematics\\
		University of Tennessee\\
		Knoxville, TN 37996\\
		USA\\

	\end{center}
	
	\vfill
	
	\begin{abstract} In this paper we investigate the effect of repulsive pair interactions on Bose--Einstein condensation in a well-established random one-dimensional system known as the Luttinger--Sy model at positive temperature. We study separately hard core interactions as well as a class of more general repulsive interactions, also allowing for a scaling of certain interaction parameters in the thermodynamic limit. As a main result, we prove in both cases that for sufficiently strong interactions all eigenstates of the non-interacting one-particle Luttinger--Sy Hamiltonian as well as any sufficiently localized one-particle state are $\mathds P$-almost surely not macroscopically occupied.
	\end{abstract}
	
	\newpage
	
	\section{Introduction}

	Bose--Einstein condensation (BEC) refers to a quantum phenomenon in bosonic many-particle systems that occurs at low temperature or high particle density. Originally, this new phenomenon was predicted by Einstein to occur in a non-interacting Bose gas in three dimensions \cite{PaperEinstein1,PaperEinstein2}. The key feature of (conventional) BEC is the so-called macroscopic occupation of a one-particle state, which means that the number of particles occupying the same one-particle state grows proportional to the volume of the system. In a non-interacting Bose gas, the one-particle states that might be macroscopically occupied are simply the eigenstates of the underlying one-particle Hamiltonian. In an interacting many-particle system, the situation is much more complex since one has to study the eigenstates of the reduced one-particle density matrix which is generally hard to construct \cite{PO56}. Consequently, it is not surprising that the first rigorous proof of BEC in an interacting system was established not long ago \cite{lieb2002proof,Lieetal05}. Also, the study of BEC in interacting systems is still a very active area in mathematical physics, see \cite{SeiDeuchYngvason,SeiDeuch,SchleinBEC} and references therein. 
	
	It is well-known that BEC in one dimension differs significantly from BEC in three and higher dimensions. For example, the non-interacting one-dimensional Bose gas (with standard boundary conditions) exhibits no BEC \cite{VerbeureBook}. However, as realized by Luttinger and Sy some fifty years ago \cite{luttinger1973bose}, BEC is recovered in the presence of a random external potential, see also \cite{kac1973bose,kac1974bose}. Loosely speaking, in the Luttinger--Sy model (LS model) the random external potential consists of a sum of Dirac-$\delta$ potentials of infinite interaction strength situated at random points distributed along $\mathds{R}$. Consequently, the real line is almost surely dissected into a countable number of intervals of random lengths on which the one-particle Hamiltonian simply acts as the Dirichlet Laplacian. It is interesting to remark that the LS model, despite the singular nature of the external potential, is considered to be a good approximation of more realistic random external potentials \cite[p.~11]{VeniaminovKlopp}, \cite[p.~8]{LenobleZagrebnovLuttingerSy} and thus plays an important role in the study of BEC in a random environment \cite{luttinger1973bose,lenoble2004bose,LenobleZagrebnovLuttingerSy,SeiYngZag12,KPS182}. On a rigorous level, BEC in the non-interacting LS model was first proved in \cite{LenobleZagrebnovLuttingerSy}.
	
	In this paper, we study existence of BEC in a system of interacting particles in $\mathds{R}$ and with a (singular) random external potential generated through a Poisson point process in the canonical ensemble at positive temperature in the thermodynamic limit. In other words, without interparticle interactions, the underlying model we consider is simply the LS model. Generalizing the ideas and methods developed in \cite{PS86,PuleAonghusa87,BolteKernerInstability}, it is our aim to investigate if the ground state of the non-interacting one-particle Luttinger--Sy Hamiltonian (or its other eigenstates for that matter) can be macroscopically occupied after one introduces repulsive pair interactions. Intuitively, one indeed expects that the introduction of repulsive interactions lowers the density of particles occupying a given state \cite{LVZ03}. However, as we will see, the surprising observation is that an introduction of a sufficiently strong (which nevertheless might be arbitrarily small in $L^1$ norm) repulsive pair interaction between the particles prohibits a macroscopic occupation of these one-particle eigenstates (and, more general, of any sufficiently localized one-particle state) in the LS model. In other words, up to the fact that we will work in the canonical rather than in the grand canonical ensemble, we prove that the BEC that is present in the non-interacting LS model is destroyed through the introduction of sufficiently strong repulsive pair interactions.
	
	The paper is organized as follows: In Section~\ref{SectionBasics} we introduce the non-interacting LS model and collect results which are important in the sequel. In Section~\ref{SectionHC} and~\ref{SectionGeneral} we consider hard core pair interactions and a more general class of repulsive pair interactions, respectively, and prove in either case that almost surely no eigenstate of the non-interacting one-particle Luttinger--Sy Hamiltonian is macroscopically occupied at positive temperatures. In addition, we mention an extension of the results to arbitrary one-particle states that are sufficiently localized. Finally, we refer to the appendix, Section~\ref{Appendix}, for an auxiliary result. 
	
	\section{The non-interacting LS model}\label{SectionBasics}
	
	In this section we introduce the non-interacting Luttinger--Sy model (LS model) and also state some of its main properties that we use in the following sections.
 Since the LS model is a random model, there is an underlying probability space which we shall denote by  $(\Omega,\mathcal{A},\mathds{P})$. 
 
 Let $\Lambda_{N}=(-L_N/2,+L_N/2) \subset \mathds{R}$ be an interval on which we place $N$ (bosonic) particles. Working in the canonical ensemble at finite inverse temperature $\beta \in (0,\infty)$, we require that
 \begin{equation}
 \frac{N}{L_N}=:\rho
 \end{equation}
 holds for all values $N \in \mathds{N}$. Here, $\rho > 0$ denotes the particle density. The thermodynamic limit is then obtained as the limit $N \rightarrow \infty$.
 
 In the non-interacting LS model, the $N$-particle Hamiltonian is informally given by
 \begin{equation}\label{NParticleHamiltonianLS}
 \begin{split}
 H^0_{N}(\omega):&=\sum_{j=1}^{N}\left(-\frac{\partial^2}{\partial x_j^2}+V(\omega,x_j)\right)
 \end{split}
 \end{equation}
 where $V(\omega,\cdot)$ is a singular random external potential defined as
 \begin{equation} \label{random external Potential definition}
 V(\omega,x):=\gamma\sum_{j \in \mathds{Z}}\delta(x-\xjomega{j}) \quad \text{ with } \quad \gamma=\infty \ .
 \end{equation}
 Here, $\delta(\cdot)$ is the Dirac-$\delta$ distribution. Furthermore, $(\xjomega{j})_{j \in \mathds{Z}}$ is a random sequence of points generated by the Poisson point process on $\mathds R$ with intensity $\nu > 0$. This means that the probability of finding $n$ such points in a Borel set $A \subset \mathds{R}$ with measure $|A|$ is given by $\frac{\nu |A|^n}{n!}\mathrm{e}^{-\nu |A|}$. In addition, the probability of finding $n$ such points in the set $A$ and $m$ points in the Borel set $B \subset \mathds{R}$ with $A \cap B = \emptyset$ is given by the corresponding product of probabilities. $\mathds{P}$-almost surely, the random points can be written as 
 \begin{equation*}
 ... < \xjomega{-1} < \xjomega{0} < 0 < \xjomega{1} < ...
  \end{equation*}
  and accordingly the real line is almost surely dissected into a countable number of intervals $\{(x_j^{\omega}, x_{j+1}^{\omega})\}_{j \in \mathds Z}$. If not stated otherwise, we always choose an $\omega \in \Omega$ with this typical property in the remainder of the paper. Also, $(\hat l^{j,\omega})_{j \in \mathds Z \backslash\{0\}}$ where $\hat l^{j,\omega}:= (\xjomega{j},\xjomega{j+1})$ are independent, identically distributed random variables with common probability density function $\nu \mathrm{e}^{-\nu l} \mathds 1_{(0,\infty)}(l), l \in \mathds R$. Here $\mathds 1_{A}$ is the characteristic function of a set $A \subset \mathds R$. For more details regarding the Poisson point process on $\mathds R$, we refer the reader to \cite[Chapter 4]{kingman1993poisson}. Lastly, we note that the informal choice $\gamma=\infty$ in \eqref{random external Potential definition} means that one imposes Dirichlet boundary conditions at each $\xjomega{j}$, $j \in \mathds Z$.
  
  Since for every $N \in \mathds N$ all $N$ bosons are confined to the interval $\Lambda_{N}$, we introduce the subintervals
  \begin{equation*}
   I_N^{j,\omega} := (\xjomega{j},\xjomega{j+1}) \cap \Lambda_N
  \end{equation*}
  as well as the (random) lengths
  \begin{equation*}
  l^{j,\omega}_N:= |I_N^{j,\omega}|= |(\xjomega{j},\xjomega{j+1}) \cap \Lambda_N|
  \end{equation*}
  and set $l^{(1),\omega}_{N,>}:=\max_{j \in \mathds{Z}}\{l^{j,\omega}_N\}$. Regarding the largest length $l^{(1),\omega}_{N,>}$ one has the following important statement; see, e.g., \cite[Theorem B.2]{SeiYngZag12} and \cite[Theorem 5.1.6]{pechmanndiss}.
  \begin{lemma}\label{LargestComponent} Let $\nu > 0$ be the intensity of the underlying Poisson point process. Then, for all $\alpha > 4$, all $0 < \varepsilon < 1$, and $\mathds{P}$-almost all $\omega \in \Omega$ there exists an $N_0 \in \mathds{N}$ such that for all $N > N_0$,  
  	\begin{equation}
  	\nu^{-1}\left[\ln(L_N)-(1+\varepsilon)\ln(\ln(L_N)) \right] \leq l^{(1),\omega}_{N,>} \leq \alpha \nu^{-1} \ln(L_N) \ .
  	\end{equation}
  \end{lemma}

From the construction above we see that the underlying one-particle Hamiltonian in the LS model is rigorously defined as a direct sum of Dirichlet Laplacians over all subintervals $I^{j,\omega}_N$ with lengths $l^{j,\omega}_N > 0$. On an informal level, the non-interacting one-particle Luttinger--Sy Hamiltonian (LS Hamiltonian) is given by $H^0_1(\omega)$ defined on $\mathrm{L}^2(\Lambda_N)$. As a consequence, the spectrum of this self-adjoint operator is purely discrete and the spectrum is the union of all eigenvalues 
\begin{equation}
\frac{\pi^2 n^2}{(l^{j,\omega}_N)^2}
\end{equation}
with $n \in \mathds{N}$ and $j \in \mathds{Z}$ for which $l^{j,\omega}_N > 0$. The ground state energy hence corresponds to the first Dirichlet eigenvalue on the interval with length $l^{(1),\omega}_{N,>}$ and it shall be denoted by $\epsilon_N^{0,\omega}$. Also, writing the eigenvalues in increasing order we write $\epsilon_N^{k,\omega}$ for the $k$-th eigenvalue, $k \in \mathds{N}_0$. The associated eigenstates shall be denoted by $\varphi^{k,\omega}_N \in \mathrm{L}^2(\Lambda_N)$, $k \in \mathds N_0$. We remark that each eigenstate $\varphi^{k,\omega}_N \in \mathrm{L}^2(\Lambda_N)$ has support on one subinterval $I^{j,\omega}_N$ only.

	\section{The case of hard core interactions}\label{SectionHC}
	 In this section we generalize the methods that were developed in \cite{PuleAonghusa87} to our random one-dimen\-sional setting. More explicitly, to the LS-model we add hard core repulsive pair interactions between the particles and investigate their effect on a macroscopic occupation of the one-particle eigenstates $\varphi^{k,\omega}_N \in \mathrm{L}^2(\Lambda_N)$.
	 
	 Introducing pair interactions between the particles, the $N$-particle Hamiltonian ($N \geq 2$) reads as
	\begin{equation}\label{NParticleHamiltonian}
	\begin{split}
	H_{N}(\omega):&=H^0_{N}(\omega)+\sum_{1 \leq i < j \leq N}U_N(x_i-x_j)
	\end{split}
	\end{equation}
	where $U_N:\mathds{R} \rightarrow [0,\infty]$ is a sequence of functions describing the interaction between two particles. Since we consider hard core pair interactions (imagining each particle as a hard one-dimensional ``sphere'') in this section, we define
	\begin{equation}\label{HardCoreIP}
	U_N(x):=
	\begin{cases}
	\infty \quad \text{for} \quad |x| \leq a_N \ ,\\
	0 \quad \text{else}\ ,
	\end{cases}
	\end{equation}
	for some $a_N > 0$ which describes the radius of the ``sphere''. We assume that the sequence $(a_N)_{N \in \mathds{N}}$ is bounded from above but it is allowed to converge to zero at a certain rate as specified below.
 
	The Hamiltonian \eqref{NParticleHamiltonian} with hard core interaction potential \eqref{HardCoreIP} is rigorously defined as suitable Dirichlet Laplacian. Indeed, the (random) $N$-particle configuration space is given by
	\begin{equation}\label{Domain}
	\Lambda^{\omega}_{N,L_N}:=\left\{y \in \Lambda_{N}^N \setminus \bigcup_{j=1}^\infty W(\xjomega{j}): |y_j-y_i| > a_N, i \neq j, i,j=1,...,N \right\}\ ,
	\end{equation}
	where $y=(y_1,y_2,...,y_N)^T \in \mathds{R}^N$ and
	\begin{equation*}
 W(\xjomega{j}):=\left\{y \in \mathds{R}^N:\ \exists i \in (1,2,...,N) \ \text{such that}\  y_i=\xjomega{j} \right\}
	\end{equation*}
	is a hyperplane associated with the $j$-th random point $\xjomega{j} \in \mathds{R}$. As a consequence, the Hamiltonian \eqref{NParticleHamiltonian} is rigorously defined as a Dirichlet Laplacian on $\mathrm{L}^2_{sym}(\Lambda^{\omega}_{N,L_N})$, the subscript indicating that we work on the fully symmetric subspace of $\mathrm{L}^2(\Lambda^{\omega}_{N,L_N})$ (recall that we are dealing with bosons). We note that this operator, being a Dirichlet Laplacian on a bounded domain, has compact resolvent and hence purely discrete spectrum only. 
	
By construction, it is clear that the particle density $\rho$ cannot --- in the case of hard core interactions --- attain any real value since each particle occupies a space of ``volume'' $2a_N > 0$. In other words, we may introduce the critical particle density
\begin{equation}\label{CriticalDensity}
\rho_{crit}:=\frac{1}{2\|a_N\|_{\infty}}
\end{equation}
and restrict attention to values $\rho < \rho_{crit}$ in the rest of this section.

	Since we study BEC in the canonical ensemble at inverse temperature $\beta \in (0,\infty)$, the $N$-particle state of system (meaning the density matrix) is given by
	\begin{equation*}
	\DensityMatrixNbeta=\frac{\mathrm{e}^{-\beta H_{N}(\omega)}}{\tr(\mathrm{e}^{-\beta H_{N}(\omega)})}\ ,
	\end{equation*}
	where $\tr(\cdot)$ refers to the trace of a (trace-class) operator on the Hilbert space $\mathrm{L}^2_{sym}(\Lambda^{\omega}_{N,L_N})$. Let $\DensityMatrixNbeta(\cdot,\cdot)$ denote the kernel of $\DensityMatrixNbeta$. In order to calculate the density of particles in a given one-particle state one makes use of the so-called reduced one-particle density matrix which acts as a trace-class operator on the underlying one-particle Hilbert space $\mathrm{L}^2(\Lambda_N)$ \cite{M07}. The kernel of the corresponding reduced one-particle density matrix is then obtained as 
	\begin{equation}
	\DensityMatrixNbetaeins(x,y)=N \int \ud z_1 ... \int \ud z_{N-1}\ \DensityMatrixNbeta(x,z_1,...,z_{N-1},y,z_1,...,z_{N-1})\ 
	\end{equation}
	with $x,y \in \Lambda_{N} \setminus (\xjomega{j})_{j \in \mathds{Z}}$. Consequently, the average particle density in a one-particle state $\varphi \in \mathrm{L}^2(\Lambda_{N})$ can be calculated as
	\begin{equation}
	\AverageParticleDensityNbeta(\varphi) := \frac{1}{L_N}\int_{\Lambda_N}\ud x \int_{\Lambda_N}\ud y\ \overline{\varphi(x)}\DensityMatrixNbetaeins(x,y) \varphi(y)\ ,
	\end{equation} 
see \cite{PuleAonghusa87,M07}. This leads to the following definition. 
\begin{defn}\label{DefMacI} Let $(\varphi_{N})_{N \in \mathds N}$, $\varphi_N \in \mathrm{L}^2(\Lambda_{N})$ for all $N \in \mathds N$, be a sequence of normalized one-particle states. We call $(\varphi_{N})_{N \in \mathds N}$ $\mathds{P}$-almost surely macroscopically occupied at inverse temperature $\beta \in (0,\infty)$ iff
	\begin{equation}
	\limsup_{N \rightarrow \infty} \AverageParticleDensityNbeta(\varphi_N) > 0
	\end{equation}
	holds for $\mathds P$-almost all $\omega \in \Omega$.
\end{defn}
Now, following \cite{PuleAonghusa87} we decompose $\mathds{R}$ into a countable number of intervals. Namely, for $n \in \mathds{Z}$ we introduce
	\begin{equation*}
	\Lambda^{(n)}_N:=\{x \in \mathds{R}: a_Nn \leq x \leq a_N(n+1) \}  \ .
	\end{equation*}
Since every particle is modeled as a hard ``sphere'' of radius $a_N > 0$, at most one particle can occupy the box $\Lambda^{(n)}_N$. 

	%
	\begin{lemma}[\cite{PuleAonghusa87}, Lemma~2]\label{LemmaAP} For any $N \in \mathds{N}$, let $\varphi \in \mathrm{L}^2(\Lambda_{N})$ be a normalized one-particle state extended by zero to all of $\mathds{R}$. Then, $\mathds{P}$-almost surely and for all $\beta \in (0,\infty)$ one has 
		\begin{equation}
		\AverageParticleDensityNbeta(\varphi) \leq \frac{1}{L_N}\left(\sum_{n \in \mathds{Z}} \left(\int_{\Lambda^{(n)}_N}|\varphi(x)|^2\ \ud x\right)^{1/2} \right)^2\ .
		\end{equation}
	\end{lemma}

Using Lemma~\ref{LargestComponent} we can now prove a first main result, i.e., we prove that $\mathds P$-almost surely none of the eigenstates of the non-interacting one-particle LS Hamiltonian are macroscopically occupied in the presence of sufficiently extended hard core interactions.
\begin{theorem}\label{MainResultBEC}\label{ProofEigenstates} For arbitrary $k \in \mathds N_0$, let $\varphi^{k,\omega}_N \in \mathrm{L}^2(\Lambda_{N})$ denote the normalized $k$-th eigenstate of the one-particle LS Hamiltonian defined on $L^2(\Lambda_N)$. Then, for all $\beta \in (0,\infty)$ and assuming that the sequence of radii $(a_N)_{N \in \mathds{N}}$ is such that
	\begin{equation*}
	\frac{\ln^2(N)}{a^2_N N} \longrightarrow 0\ , \quad \text{as}\quad N \rightarrow \infty \ ,
	\end{equation*}
	one has
	\begin{equation}
	\lim_{N \rightarrow \infty} \AverageParticleDensityNbeta(\varphi^{k,\omega}_N)=0 
	\end{equation}
	for $\mathds{P}$-almost all $\omega \in \Omega$. In other words, the sequence $(\varphi^{k,\omega}_N)_{N \in \mathds N}$ is $\mathds P$-almost surely not macroscopically occupied.
\end{theorem}
\begin{proof}
    Recall from Section~\ref{SectionBasics} that any eigenstate $\varphi^{k,\omega}_N$ has support on exactly one subinterval $I^{j,\omega}_N$.
    With Lemma~\ref{LemmaAP} and Lemma~\ref{LargestComponent} we obtain
	\begin{equation*}\begin{split}
		\AverageParticleDensityNbeta(\varphi^{k,\omega}_N) &\leq \frac{1}{L_N}\left(\sum_{n \in \mathds{Z}}\left(\int_{\Lambda^{(n)}_N}|\varphi^{k,\omega}_N(x)|^2\ \ud x\right)^{1/2}\right)^2 \\
&\leq \frac{1}{L_N}\left(\sum_{n \in \mathds{Z}: \supp(\varphi^{k,\omega}_N)\cap \Lambda^{(n)}_N \neq \emptyset}1\right)^2 \\
& \leq \frac{\alpha^2 \nu^{-2}\ln^2(L_N)}{a^2_NL_N} 
		\end{split}
	\end{equation*}
	for $N$ large enough and with some $\alpha > 4$. The statement now follows readily.
\end{proof}
\begin{remark}\label{RemarkLocalI}
	Under the assumptions of Theorem~\ref{ProofEigenstates} the following holds: From the proof of Theorem~\ref{ProofEigenstates} one concludes that any sequence of normalized one-particle states $(\varphi_N)_{N \in \mathds N}$, $\varphi_N \in \mathrm{L}^2(\Lambda_{N})$ for each $N \in \mathds N$, cannot be macroscopically occupied given that \linebreak $\lim_{N \to \infty} S_N^2/N =0$ holds where $S_N$ denotes the number of boxes $\Lambda^{(n)}_N$ on which $\varphi_N$ is supported.
	
	 Alternatively, a sequence of normalized one-particle states $(\varphi_{N})_{N \in \mathds N}$, $\varphi_{N} \in \mathrm{L}^2(\Lambda_{N})$ for each $N \in \mathds N$, is not macroscopically occupied if there exists a sequence $(x_N)_{N \in \mathds{R}} \subset \mathds{R}$ such that
	\begin{equation*}
	|\varphi_{N}(x)| \leq \frac{C}{|x-x_N|^{1+\varepsilon}}\ 
	\end{equation*}
	holds for some constants $0 <\varepsilon \leq 1/2$, $C > 0$ (independent of $N$), all $N \in \mathds{N}$ and almost all $x \in \mathds{R}$ with $|x-x_N| > R$ and $R$ independent of $N$.
	
\end{remark}
	\section{The case of general repulsive interactions}\label{SectionGeneral}
	This section is based on the methods developed first in \cite{PS86} and generalized later in \cite{BolteKernerInstability}. Although strictly speaking not necessary, it is convenient to utilize methods of second quantization in this section, see e.g. \cite{MR04} for an introduction.

	In contrast to the previous section we allow for more general and less repulsive pair interactions in the sequel. Namely, in \eqref{NParticleHamiltonian} we now allow for a sequence of potentials $U_N(\cdot) \in (\mathrm{L}^{\infty}\cap \mathrm{L}^1)(\mathds{R})$, $ 0 \leq U_N(x) < \infty$, $\|U_N\|_{\mathrm{L}^1(\mathds{R})}$ uniformly bounded from above, and that fulfills the following property: For all $N \in \mathds N$,
	\begin{equation}\label{weakCond}
\exists A_N,b_N > 0: \forall x \in [-A_N,+A_N]\ \text{one has}\ U_N(x) \geq b_N \ .
	\end{equation}
The condition~\eqref{weakCond} can be interpreted as a weak version of the hard core interactions considered in the previous section which correspond to the informal choice $A_N:=a_N$ and $b_N:=\infty$. As will become clear later, the sequence $(A_N)_{N \in \mathds{N}}$ shall be bounded from above but it will be allowed to converge to zero at some specified rate. As for the sequence $(b_N)_{N \in \mathds{N}}$ we will allow for two different cases, one in which this sequence converges to zero and one in which it is bounded away from zero.

The $N$-particle Hilbert space is given as $\mathrm{L}^2_{sym}(\widetilde{\Lambda}^{\omega}_{N,L_N})$ with
\begin{equation}\label{Domain2}
\widetilde{\Lambda}^{\omega}_{N,L_N}:=\left\{y \in \Lambda_{N}^N \setminus \bigcup_{j \in \mathds{Z}} W(\xjomega{j}) \right\}
\end{equation}
and $W(\xjomega{j})$ as defined below \eqref{Domain}. For a suitable $N$-particle operator $\mathcal O_N$ on the Hilbert space $\mathrm{L}^2_{sym}(\widetilde{\Lambda}^{\omega}_{N,L_N})$, its expectation value in the canonical ensemble is given by
\begin{equation}
\EnsembleExpectation{\mathcal O_N} :=\tr(\mathcal O_N \DensityMatrixNbeta)\ .
\end{equation}
In a first step we now prove that the energy density in the interacting system remains bounded in the thermodynamic limit. 
\begin{lemma}\label{EnergyDensity} Consider the $N$-particle Hamiltonian~\eqref{NParticleHamiltonian} on the Hilbert space $\mathrm{L}^2_{sym}(\widetilde{\Lambda}^{\omega}_{N,L_N})$ with an interaction potential $U_N$ as defined above. Then, for all $\beta \in (0,\infty)$ and $\mathds{P}$-almost surely one has
	\begin{equation*}
	\limsup_{N \rightarrow \infty} \frac{\EnsembleExpectation{H_N(\omega)}}{L_N} < \infty\ .
	\end{equation*}
\end{lemma}
\begin{proof}
  Let $\beta \in (0,\infty)$ and $N \in \mathds N$ (and a typical $\omega \in \Omega$ from a set of full measure) be given. Note that
  \begin{equation*}
   \EnsembleExpectation{H_N(\omega)}= - \left( \dfrac{ \mathrm{d}}{\mathrm{d} \tilde \beta} \mathscr Z_N(\tilde \beta) \right)_{\tilde \beta = \beta}
  \end{equation*}
  where
  \begin{equation*}
    \mathscr Z_N : (0,\infty) \to \mathds R, \ \tilde \beta \mapsto \mathscr Z_N(\tilde \beta) := \ln(\tr(\mathrm{e}^{-\tilde \beta H_N(\omega)})) \ .
  \end{equation*}
   Since $\mathscr Z_N$ is a convex function,
  \begin{equation*}
   \EnsembleExpectation{H_N(\omega)} \le \dfrac{\mathscr Z_N(\beta/2) - \mathscr Z_N(\beta)}{\beta/2} \ .
  \end{equation*}
On the one hand, with $(E_N^{j})_{j \in \mathds N_0}$ being the eigenvalues of the Dirichlet Laplacian on $\Lambda_N^N$, arranged in increasing order and repeated according to their multiplicity, we have
\begin{equation*}
 \mathscr Z_N(\beta/2) \le \ln\left(\sum\limits_{j \in \mathds N_0} \mathrm{e}^{-(\beta/2) E_N^{j}}\right) 
\end{equation*}
where the right hand side, after dividing by $L_N$, converges to a constant $C_{1}(\beta/2) > 0$ in the limit $N \to \infty$, see, e.g., \cite[Theorem~3.5.8]{ruelle1999statistical}. 

To give a lower bound, we define the trial function
$$\Psi_N^{\omega}(x_1,\ldots,x_N) := \psi_N^{\omega}(x_1) \ldots \psi_N^{\omega}(x_N) $$
where $\psi_N^{\omega}(x) = \sum_{j \in \mathds Z} \xi_N^{j,\omega}(x) \mathds 1_{l_N^{j,\omega}\ge 3}(x)$.
Here, for all $j \in \mathds Z$ such that $l_N^{j,\omega} \ge 3$, we define a function $\xi_N^{j,\omega} : I_N^{j,\omega} \to \mathds R$ on the associated subinterval $I_N^{j,\omega}$ as
\begin{equation*}
 x \mapsto \xi_N^{j,\omega}(x) := \begin{cases}
                                                                                \dfrac{\mathrm{e}^{-(x-\xjomega{j})^{-1}}}{\mathrm{e}^{-(x - \xjomega{j})^{-1}} + \mathrm{e}^{-(x - \xjomega{j} - 1)^{-1}}} & \text{ if } x \in (\xjomega{j}, \xjomega{j} + 1) \vspace{0.2in}\ ,\\
                                                                                1 & \text{ if } x \in [\xjomega{j}+1 , \hat x^{j+1,\omega} -1] \vspace{0.2in}\ , \\
                                                                                \dfrac{\mathrm{e}^{-(x-\hat x^{j + 1,\omega})^{-1}}}{\mathrm{e}^{-(x - \hat x^{j + 1,\omega})^{-1}} + \mathrm{e}^{-(x - \hat x^{j + 1,\omega} + 1)^{-1}}} & \text{ if } x \in (\hat x^{j+1,\omega} - 1, \hat x^{j+1,\omega})\ .
                                                                               \end{cases} 
\end{equation*}
Note that $\|\Psi_N^{\omega}\|_2^{-1} \Psi_N^{\omega} \in \mathcal D(H_N(\omega))$, where $\mathcal{D}(H_N(\omega))$ is the domain of the self-adjoint operator $H_N(\omega)$. We obtain
\begin{align*}
 \mathscr Z_N(\beta) &\ge \ln \left( \mathrm{e}^{-\beta \| \Psi_N^{\omega}\|_2^{-2} \langle \Psi_N^{\omega}, H_N(\omega) \Psi_N^{\omega} \rangle} \right) = -\beta \| \Psi_N^{\omega}\|_2^{-2} \langle \Psi_N^{\omega}, H_N(\omega) \Psi_N^{\omega} \rangle
\end{align*}
due to the convexity of the exponential function and the Jensen's inequality, see also \cite[Lemma 14.1 and Remark 14.2]{LiebSeiringer}, as well as the monotonicity of the natural logarithm.

Employing Lemma~\ref{Lemma Anzahl Intervalle mit Mindestlaenge 3} and the fact that $\|\psi_N^{\omega}\|_2^2 \ge \#\{j \in \mathds Z: l_N^{j,\omega} \ge 3\}$ (here, $\#A$ denotes the number of elements in a set $A \subset \mathds N_0$), we obtain
\begin{align*}
 & \|\Psi_N^{\omega}\|_2^{-2} \langle \Psi_N^{\omega}, H_N(\omega) \Psi_N^{\omega} \rangle \\
 & \quad = \, N \|\psi_N^{\omega}\|_2^{-2} \int\limits_{\Lambda_N} |(\psi_N^{\omega})'(x)|^2 \, \mathrm{d} x \\
 & \qquad \qquad + \, \dfrac{N(N-1)}{2} \|\psi_N^{\omega}\|_2^{-4} \int\limits_{\Lambda_N} \int\limits_{\Lambda_N} U_N(x-y) |\psi_N^{\omega}(x)|^2 |\psi_N^{\omega}(y)|^2 \, \mathrm{d} x \, \mathrm{d} y \\
 & \quad \le \, N (const.) + \dfrac{N(N-1)}{2} \dfrac{1}{(\#\{j \in \mathds Z: l_N^{j,\omega} \ge 3\})^2} L_N \|U_N\|_{\mathrm{L}^1(\mathds R)} \\
 & \quad \le \, N C_2
\end{align*}
for a constant $C_2>0$ independent of $N$. 

Thus, in summary we get
\begin{equation*}
 \limsup\limits_{N \to \infty} \dfrac{\EnsembleExpectation{H_N(\omega)}}{N}  \le \dfrac{2}{\beta} \left( C_1 \left( \dfrac{\beta}{2} \right) + \beta C_2 \right) < \infty \ .
\end{equation*}
\end{proof}
\begin{remark} It is interesting to remark that questions related to the convergence of the energy density in the interacting LS model where recently studied in \cite{Veniaminov,VeniaminovKlopp}.
\end{remark}
In order to prove the main theorem of this section, we employ methods from second quantization. For a given one-particle state $\varphi \in \mathrm{L}^2(\Lambda_N)$, we write $a^{\ast}(\varphi)$ for the corresponding creation and $a(\varphi)$ for the corresponding annihilation operator. As usual, they can be expressed in terms of the distribution valued operators $a^{\ast}(x)$ and $a(x)$ which fulfill the CCR
\begin{equation}\begin{split}
\left[a(x),a^{\ast}(y)\right]&=\delta(x-y)\ , \\
\left[a(x),a(y)\right]&=\left[a^{\ast}(x),a^{\ast}(y)\right]=0 \ ,
\end{split}
\end{equation}
where $x,y \in \Lambda_N \setminus \{\xjomega{j}: j \in \mathds{Z}\}$. One has
\begin{equation*}\begin{split}
a^{\ast}(\varphi)&=\int_{-L_N/2}^{+L_N/2} \varphi(x)a^{\ast}(x)\ \ud x 
\end{split}
\end{equation*}
as well as
\begin{equation*}\begin{split}
a(\varphi)&=\int_{-L_N/2}^{+L_N/2}\overline{\varphi(x)}a(x)\ \ud x \ .
\end{split}
\end{equation*}
To a given state $\varphi \in \mathrm{L}^2(\Lambda_N)$ one associates the number operator $a^{\ast}(\varphi)a(\varphi)$. 
\begin{remark}\label{DefMacII} We note that Definition~\ref{DefMacI} is translated to the present setting as follows: Let $\varphi_{N} \in \mathrm{L}^2(\Lambda_N)$ be a normalized one-particle state for all $N \in \mathds N$. We say that $(\varphi_N)_{N \in \mathds N}$ is $\mathds{P}$-almost surely macroscopically occupied at inverse temperature $\beta \in (0,\infty)$ iff
	\begin{equation}
	\limsup_{N \rightarrow \infty}\frac{\EnsembleExpectation{a^\ast(\varphi_{N})a(\varphi_{N})}}{L_N} > 0 
	\end{equation}
	holds for $\mathds P$-almost all $\omega \in \Omega$.
\end{remark}

We are now in position to prove the main result of this section, namely, that $\mathds P$-almost surely all eigenstates of the LS Hamiltonian on $L^2(\Lambda_N)$ are not macroscopically occupied given the two-particle interactions are sufficiently strong.
\begin{theorem}\label{CondensateXXX} For arbitrary $k \in \mathds N_0$, let $\varphi^{k,\omega}_{N} \in \mathrm{L}^2(\Lambda_{N})$ denote the normalized $k$-th eigenstate of the one-particle LS Hamiltonian defined on $L^2(\Lambda_N)$ for all $N \in \mathds{N}$. Then, for all $\beta \in (0,\infty)$, assuming that either the bounded sequence $(A_N)_{N \in \mathds{N}}$ is such that 
	\begin{equation*}
	\frac{A^3_N N}{\ln^3(N)} \longrightarrow \infty\ ,\quad \text{as}\ N \longrightarrow \infty \ ,
	\end{equation*}
and $(b_N)_{N \in \mathds{N}}$ is bounded away from zero, or that $(b_N)_{N \in \mathds{N}}$ converges to zero and
\begin{equation*}
\frac{b_N A^3_N N}{\ln^3(N)} \longrightarrow \infty\ ,\quad \text{as}\ N \longrightarrow \infty\ ,
\end{equation*}
we have that $(\varphi^{k,\omega}_{N})_{N \in \mathds N}$ is $\mathds{P}$-almost surely not macroscopically occupied, i.e.,
	\begin{equation}
	\limsup_{N \rightarrow \infty}\frac{\EnsembleExpectation{a^\ast(\varphi^{k,\omega}_{N})a(\varphi^{k,\omega}_{N}}}{L_N}=0
	\end{equation} %
	for $\mathds P$-almost all $\omega \in \Omega$.
\end{theorem}
\begin{proof} We prove the statement for the one-particle ground state $\varphi^{0,\omega}_{N}$ only, the generalization to all other eigenstates being obvious. The proof follows from a suitable adaptation of the proof of Proposition~1 in \cite{PS86}. The basic idea is to show that, under the assumptions of Theorem~\ref{CondensateXXX} and given that there was a set $\widetilde \Omega \subset \Omega$ with $\mathds P(\widetilde \Omega) > 0$ such that $(\varphi^{0,\omega}_{N})_{N \in \mathds N}$ is macroscopically occupied for all $\omega \in \widetilde \Omega$, then for all $\omega \in \widetilde \Omega$ the energy density would go to infinity along a subsequence, in contradiction with Lemma~\ref{EnergyDensity}.
	
	Hence, suppose there is a set $\widetilde \Omega \subset \Omega$ with $\mathds P(\widetilde \Omega) > 0$ and such that $(\varphi^{0,\omega}_{N})_{N \in \mathds N}$ is macroscopically occupied for all $\omega \in \widetilde \Omega$. Let $\omega \in \widetilde \Omega$ be given, and let
	$z_N^{0,\omega} \in \Lambda_{N}$ denote the starting point of the subinterval on which the ground state is supported. We start with the basic estimate
	\begin{equation*}\begin{split}
	\frac{\EnsembleExpectation{H_N(\omega)}}{L_N} &\geq \epsilon_N^{0,\omega} \frac{N}{L_N} \\
	& \quad +\frac{1}{2L_N}\int_{-L_N/2}^{+L_N/2} \ud x \int_{-L_N/2}^{+L_N/2} \ud y \ U(x-y) \EnsembleExpectation{a^{\ast}(x)a^{\ast}(y)a(x)a(y)}\\ 
	& \geq \epsilon_N^{0,\omega} \frac{N}{L_N}\\
	+\frac{b_N}{2L_N}\sum_{j=1}^{\lceil l^{(1),\omega}_{N,>}/A_N\rceil}&\int_{z_N^{0,\omega}+(j-1)A_N}^{\min\{z_N^{0,\omega}+jA_N, L_N/2\}}\ud x\int_{z_N^{0,\omega}+(j-1)A_N}^{\min\{z_N^{0,\omega}+jA_N, L_N/2\}}\ud y \ \EnsembleExpectation{a^{\ast}(x)a^{\ast}(y)a(x)a(y)}\ .
	\end{split}
	\end{equation*}
	Since one has $\epsilon_N^{0,\omega}= \pi^2/ (l^{(1),\omega}_{N,>})^2$, the first term converges to zero by Lemma~\ref{LargestComponent}. The rest of the proof then consists in finding a suitable lower bound for the term
	\begin{equation*}\begin{split}
	\frac{b_N}{2L_N}\sum_{j=1}^{\lceil  l^{(1),\omega}_{N,>}/A_N\rceil}\int_{z_N^{0,\omega}+(j-1)A_N}^{\min\{z_N^{0,\omega}+jA_N, L_N/2\}}\ud x\int_{z_N^{0,\omega}+(j-1)A_N}^{\min\{z_N^{0,\omega}+jA_N, L_N/2\}}\ud y \ &  \EnsembleExpectation{a^{\ast}(x)a^{\ast}(y)a(x)a(y)}\\
	&:=	\frac{b_N}{2L_N}\sum_{j=1}^{\lceil  l^{(1),\omega}_{N,>}/A_N\rceil}B^{(j)}_{N}\ .
	\end{split}
	\end{equation*}
	For $j=1,...,\lceil  l^{(1),\omega}_{N,>}/A_N\rceil$ we set $\varphi^{0,\omega,(j)}_{N}:=\varphi^{0,\omega}_{N}\mathds 1_{(z_N^{0,\omega}+(j-1)A_N,z_N^{0,\omega}+jA_N)}$.
	
	Using exactly the same estimates as in eqs.~(14)-(16b) in \cite{PS86} we obtain
	\begin{equation*}
	\sum_{i,j=1}^{\lceil l^{(1),\omega}_{N,>}/A_N \rceil} \left( \EnsembleExpectation{a^{\ast}(\varphi^{0,\omega,(i)}_{N})a(\varphi^{0,\omega,(j)}_{N})} \right)^4 \leq \left(\sum_{j=1}^{\lceil l^{(1),\omega}_{N,>}/A_N \rceil}B_N^{(j)}+\rho L_N\right)^2\ .
	\end{equation*}
	Then, applying the standard estimate $\left|\sum_{i=1}^{n}x_j \right|^2 \leq n \sum_{i=1}^{n}|x_i|^2$ one concludes
\begin{equation*}\begin{split}
\frac{1}{\lceil l^{(1),\omega}_{N,>}/A_N \rceil^6}\left(\sum_{i,j=1}^{\lceil l^{(1),\omega}_{N,>}/A_N \rceil} \EnsembleExpectation{a^{\ast}(\varphi^{0,\omega,(i)}_{N})a(\varphi^{0,\omega,(j)}_{N})} \right)^4&=\frac{1}{\lceil l^{(1),\omega}_{N,>}/A_N \rceil^6} \left( \EnsembleExpectation{a^{\ast}(\varphi^{0,\omega}_{N})a(\varphi^{0,\omega}_{N})} \right)^4 \\
&\leq \left(\sum_{j=1}^{\lceil l^{(1),\omega}_{N,>}/A_N \rceil}B_N^{(j)}+\rho L_N\right)^2\ .
\end{split}
\end{equation*}
This yields, employing Lemma~\ref{LargestComponent},
\begin{equation*}\begin{split}
\frac{b_N}{2L_N}\sum_{j=1}^{\lceil l^{(1),\omega}_{N,>}/A_N\rceil} B^{(j)}_{N} &\geq \frac{b_N}{2 L_N \lceil l^{(1),\omega}_{N,>}/A_N \rceil^3} \langle a^{\ast}(\varphi^{0,\omega}_{N})a(\varphi^{0,\omega}_{N}) \rangle^{2}_{\DensityMatrixNbeta}-\frac{b_N\rho}{2} \\
&\geq b_N\left( \frac{\nu^3A^{3}_{N}N}{16 \alpha^{3} \rho \ln^3(N)}\left(\frac{ \langle a^{\ast}(\varphi^{0,\omega}_{N})a(\varphi^{0,\omega}_{N}) \rangle_{\DensityMatrixNbeta}}{L_N}\right)^2-\frac{\rho}{2}\right)\ .
\end{split}
\end{equation*}
for all $N$ large enough and some $\alpha > 4$. Now, in the first case we have that the sequence $(b_N)_{N \in \mathds{N}}$ is bounded away from zero and that the term inside the brackets diverges along a subsequence. In the second case we multiply $b_N$ into the bracket: The second term $b_N\rho/2$ then converges to zero whereas the first term diverges, again along a subsequence.
\end{proof}
\begin{remark}\label{RemarkLocalII} From the proof of Theorem~3.4 in \cite{BolteKernerInstability} it follows that Theorem~\ref{CondensateXXX} also holds for sequences of normalized one-particle states $(\varphi_{N})_{N \in \mathds N}$, $\varphi_{N} \in \mathrm{L}^2(\Lambda_{N})$ for all $N \in \mathds N$, for which there exists a sequence $(x_N)_{N \in \mathds{N}} \subset \mathds{R}$ such that
	\begin{equation*}
	|\varphi_N(x)| \leq \frac{C}{|x-x_N|^{1+\varepsilon}}\ 
	\end{equation*}
	holds for some constants $0 < \varepsilon \leq 1/2$, $C > 0$ (independent of $N$), all $N \in \mathds{N}$ and almost all $x \in \mathds{R}$ with $|x-x_N| > R$ and $R$ independent of $N$.
\end{remark}
\begin{remark}\label{Remark SYZ} In view of Remark~\ref{RemarkLocalII} it is interesting to refer to the paper \cite{SeiYngZag12} where BEC into the minimizer of the Gross--Pitaevskii functional is investigated for the LS model with two-particle interaction of the Lieb--Liniger type at zero temperature. In this paper the authors mention that, given the strength of the pair interaction is relatively large in a certain scaling limit, then the minimizer has support on all subintervals $I_N^{j,\omega} \subset \Lambda_N$. In other words, the Bose--Einstein condensate is delocalized, in agreement with Remark~\ref{RemarkLocalII}.
	
	The authors also mention that, in a scaling limit where the strength of the pair interaction converges to zero fast enough, the minimizer occupies only a small fraction of all subintervals $I_N^{j,\omega} \subset \Lambda_N$ (the authors refer to the ``localization regime''). Comparing this with Theorem~\ref{CondensateXXX} we conclude the following: Let $(\varphi_{N})_{N \in \mathds N}$, $\varphi_N \in \mathrm{L}^2(\Lambda_N)$ for all $N \in \mathds N$, be a sequence of normalized one-particle states such that $\varphi_N$ is supported only on at most $cN^{\gamma}$ subintervals $I_N^{j,\omega} \subset \Lambda_N$ where $c > 0$ is some fixed constant and $0 \leq \gamma < 1$. From the proof of Theorem~\ref{CondensateXXX} we conclude that $(\varphi_{N})_{N \in \mathds N}$ cannot be macroscopically occupied given that $(b_N)_{N \in \mathds{N}}$ is bounded away from zero and that
	\begin{equation*}
	\frac{A^3_N N^{1-3\gamma}}{\ln^3(N)} \longrightarrow \infty\ ,\quad \text{as}\ N \longrightarrow \infty .
	\end{equation*}
	%
	%
	Hence, assuming in addition that $A_N:=aN^{-\alpha}$ for some $0 <\alpha \leq 1/3$ and $a > 0$, we see that if $(\varphi_N)_{N \in \mathds N}$ is macroscopically occupied, then necessarily $\gamma \geq 1/3-\alpha$. As a consequence, the sequence states $(\varphi_{N})_{N \in \mathds N}$ cannot be too localized.

\end{remark}
\begin{remark} It is also worth to mention that we may choose a sequence of two-particle interaction potentials $(U_N)_{N \in \mathds{N}}$ that approximate a Dirac-$\delta$ contact interaction in the thermodynamic limit. Note that such an interaction is frequently studien in the context of BEC in interacting systems~\cite{Lieetal05,SeiYngZag12,KPS18}. 
	
	Indeed, one picks a non-negative function $\varphi \in L^1(\mathds{R})$ with $\int_{\mathds{R}}\varphi(x)\ \ud x =1$ and sets 
	\begin{equation*}
	U_N(x):=\frac{1}{\epsilon_N}\varphi\left(\frac{x}{\epsilon_N}\right)\ ,
	\end{equation*}
	where $(\epsilon_N)_{N \in \mathds{N}}$ is a sequence converging to zero. To be more precise, we may choose 
	\begin{equation*}
	\varphi(x):=
	\begin{cases}
	1/2 \quad \text{for} \quad |x| \leq 1 \ ,\\
	0 \quad \text{else}\ .
	\end{cases}
	\end{equation*}
	Comparing this with \eqref{weakCond}, we identify $b_N:=1/(2\epsilon_N)$ and $A_N:=\epsilon_N$. Consequently, in Theorem~\ref{CondensateXXX} the first case applies and the one-particle eigenstates of the LS Hamiltonian are not macroscopically occupied given that
	\begin{equation*}
		\frac{\epsilon^3_N N}{\ln^3(N)} \longrightarrow \infty\ ,\quad \text{as}\ N \longrightarrow \infty \ .
	\end{equation*}
	In this sense, choosing a sequence $(\epsilon_N)_{N \in \mathds{N}}$ that converges to zero not too fast, Theorem~\ref{CondensateXXX} suggests that repulsive two-particle contact interactions of the Lieb-Liniger type also destroy BEC in the ground state of the Luttinger--Sy model.
\end{remark}

%
\appendix

\section{An auxiliary result}\label{Appendix}
In this appendix we present one probabilistic result for the LS model which is used to establish Lemma~\ref{EnergyDensity}. We refer to Section~\ref{SectionBasics} for the notation used.
\begin{lemma} \label{Lemma Anzahl Intervalle mit Mindestlaenge 3}
 For $\mathds P$-almost all $\omega \in \Omega$ exists an $\widetilde N \in \mathds N$ such that for all $N \ge \widetilde N$,
 \begin{equation*}
  \#\{j \in \mathds Z : l_N^{j,\omega} \ge 3\} \ge \dfrac{\nu}{4\mathrm{e}^{3\nu}\rho} N \ .
 \end{equation*}
\end{lemma}
\begin{proof}
 For all $k \in \mathds N$ we define the set $J_k := \{ -k, -k + 1, \ldots, k - 1, k\} \backslash \{0\} \subseteq \mathds N$.
 
Let
 \begin{equation*}
  F_{2k}^{\nu,\omega}(l):= \dfrac{1}{2k} \sum\limits_{j \in J_k} \mathds 1_{\hat l^{j,\omega} < l} = \dfrac{1}{2k} \#\{ j \in J_k: \hat l^{j,\omega} < l\}
 \end{equation*}
 be the empirical distribution function with respect to the independent, identically distributed random variables $\{\hat l^{j,\omega} : j \in \mathds Z \backslash\{0\}\}$ with common probability density function $\nu \mathrm{e}^{-\nu l} \mathds 1_{(0,\infty)}(l), l \in \mathds R$.
 
By the strong law of large numbers, we $\mathds P$-almost surely have $\lim_{k \to \infty} F_{2k}^{\nu,\omega}(3) = 1 - \mathrm{e}^{-3\nu}$.
Therefore, there exists a set $\widetilde \Omega_1 \subset \Omega$ with $\mathds P(\widetilde \Omega_1) =1$ and the following property: For all $\omega \in \widetilde \Omega_1$ there exists an $\widetilde K \in \mathds N$ such that for all $k \ge \widetilde K$,
\begin{equation*}
 \#\{ j \in J_k : \hat l^{j,\omega} \ge 3\} \ge \dfrac{1}{2 \mathrm{e}^{3\nu}} 2k \ .
\end{equation*}
%


 On the other hand, there is also a set $\widetilde \Omega_2 \subset \Omega$ with $\mathds P(\widetilde \Omega_2) =1$ and the following property: For all $\omega \in \widetilde \Omega_2$ there exists a $\widetilde N \in \mathds N$ such that for all $N \ge \widetilde N$, $\kappa_{N}^{(1),\omega} \ge \lceil (1/4) \nu \rho^{-1} N \rceil + 1$ and $\kappa_{N}^{(2),\omega} \ge \lceil (1/4) \nu \rho^{-1} N \rceil + 1$,
where $\kappa_{N}^{(1),\omega}$ and $\kappa_{N}^{(2),\omega}$ are the number of the atoms of the Poisson random measure within $(0,+L_N/2)$ and $(-L_N/2,0)$, respectively. 

To conclude, for $\mathds{P}$-almost all $\omega \in \Omega$ and for all $N \ge \max\{\widetilde N, 2 \rho \nu^{-1} \widetilde K \}$ one has
\begin{equation*}
 \#\{j \in \mathds Z : l_N^{j,\omega} \ge 3\} \ge \#\{j \in J_{\lceil \nu N / (4 \rho) \rceil} : \hat l^{j,\omega} \ge 3\} \ge \dfrac{1}{2\mathrm{e}^{3\nu}} \dfrac{\nu N}{2\rho} \ .
\end{equation*}
\end{proof}

	{\small
		\bibliographystyle{amsalpha}
		\bibliography{HardCoreLS}}

\providecommand{\bysame}{\leavevmode\hbox to3em{\hrulefill}\thinspace}
\providecommand{\MR}{\relax\ifhmode\unskip\space\fi MR }
\providecommand{\MRhref}[2]{%
  \href{http://www.ams.org/mathscinet-getitem?mr=#1}{#2}
}
\providecommand{\href}[2]{#2}
\begin{thebibliography}{KPS19b}

\bibitem[ABS20]{SchleinBEC}
A.~Adhikari, C.~Brennecke, and B.~Schlein, \emph{{Bose--Einstein condensation
  beyond the Gross--Pitaevskii regime}}, arXiv:2002.03406 (2020).

\bibitem[AP87]{PuleAonghusa87}
P.~M. Aonghusa and J.~V. Pul{\'e}, \emph{{Hard cores destroy {B}ose--{E}instein
  condensation}}, Lett.~Math.~Phys. \textbf{14} (1987), no.~2, 117--121.

\bibitem[BK16]{BolteKernerInstability}
J.~Bolte and J.~Kerner, \emph{{Instability of {B}ose--{E}instein condensation
  into the one-particle ground state on quantum graphs under repulsive
  perturbations}}, J.~Math.~Phys. \textbf{57} (2016), 043301.

\bibitem[DS20]{SeiDeuch}
A.~Deuchert and R.~Seiringer, \emph{{Gross--{P}itaevskii limit of a homogeneous
  {B}ose gas at positive temperature}}, Archive for Rational Mechanics and
  Analysis \textbf{236} (2020), 1217–1271.

\bibitem[DSY19]{SeiDeuchYngvason}
A.~Deuchert, R.~Seiringer, and J.~Yngvason, \emph{{Bose--Einstein Condensation
  in a dilute, trapped gas at positive temperature}}, Comm.~Math.~Phys.
  \textbf{368} (2019), 723–776.

\bibitem[Ein24]{PaperEinstein1}
A.~Einstein, \emph{{Quantentheorie des einatomigen idealen Gases}}, Sitzber.
  Kgl. Preuss. Akad. Wiss. (1924), 261--267.

\bibitem[Ein25]{PaperEinstein2}
\bysame, \emph{{Quantentheorie des einatomigen idealen Gases}}, II. Abhandlung,
  Sitzber. Kgl. Preuss. Akad. Wiss. (1925), 3--14.

\bibitem[Kin93]{kingman1993poisson}
J.~F.~C. Kingman, \emph{{Poisson processes}}, Clarendon Press, 1993.

\bibitem[KL73]{kac1973bose}
M.~Kac and J.~M. Luttinger, \emph{{Bose--Einstein condensation in the presence
  of impurities}}, J. Math. Phys. \textbf{14} (1973), 1626--1628.

\bibitem[KL74]{kac1974bose}
\bysame, \emph{{Bose--Einstein condensation in the presence of impurities.
  II}}, J. Math. Phys. \textbf{15} (1974), 183--186.

\bibitem[KPS19a]{KPS18}
J.~Kerner, M.~Pechmann, and W.~Spitzer, \emph{{Bose--Einstein condensation in
  the Luttinger--Sy model with contact interaction}}, Ann. Henri Poincar{\'e}
  \textbf{20} (2019), 2101--2134.

\bibitem[KPS19b]{KPS182}
\bysame, \emph{{On Bose--Einstein condensation in the Luttinger--Sy model with
  finite interaction strength}}, J. Stat. Phys. \textbf{174} (2019),
  1346--1371.

\bibitem[KV14]{VeniaminovKlopp}
F.~Klopp and N.~A. Veniaminov, \emph{{Interacting electrons in a random medium:
  a simple one-dimensional model}}, arxiv:1408.5839 (2014).

\bibitem[LPZ04]{lenoble2004bose}
O.~Lenoble, L.~A. Pastur, and V.~A. Zagrebnov, \emph{{Bose--Einstein
  condensation in random potentials}}, Comptes Rendus Physique \textbf{5}
  (2004), 129--142.

\bibitem[LS73]{luttinger1973bose}
J.~M. Luttinger and H.~K. Sy, \emph{{Bose--Einstein condensation in a
  one-dimensional model with random impurities}}, Phys. Rev.~A \textbf{7}
  (1973), 712--720.

\bibitem[LS02]{lieb2002proof}
E.~H. Lieb and R.~Seiringer, \emph{{Proof of Bose--Einstein condensation for
  dilute trapped gases}}, Phys. Rev. Lett. \textbf{88} (2002), 170409.

\bibitem[LS10]{LiebSeiringer}
\bysame, \emph{{The stability of matter in quantum mechanics}}, Cambridge
  University Press, 2010.

\bibitem[LSSY05]{Lieetal05}
E.~H. Lieb, R.~Seiringer, J.~P. Solovej, and J.~Yngvason, \emph{{The
  mathematics of the {B}ose gas and its condensation}}, {Oberwolfach Seminars},
  vol.~34, Birkh{\"a}user Verlag, Basel, 2005.

\bibitem[LVZ03]{LVZ03}
J.~Lauwers, A.~Verbeure, and V.~A. Zagrebnov, \emph{{Proof of
  {B}ose--{E}instein condensation for interacting gases with a one-particle
  gap}}, J.~Phys.~A \textbf{36} (2003), 169--174.

\bibitem[LZ07]{LenobleZagrebnovLuttingerSy}
O.~Lenoble and V.~A. Zagrebnov, \emph{{Bose--{E}instein condensation in the
  {L}uttinger--{S}y model}}, Mark.~Proc.~Rel.~Fields \textbf{13} (2007),
  441--468.

\bibitem[Mic07]{M07}
A.~Michelangeli, \emph{{Reduced density matrices and Bose--Einstein
  condensation}}, SISSA \textbf{39} (2007).

\bibitem[MR04]{MR04}
P.~A. Martin and F.~Rothen, \emph{{Many-body problems and quantum field
  theory}}, Springer-{V}erlag, 2004.

\bibitem[Pec19]{pechmanndiss}
M.~Pechmann, \emph{{Bose--Einstein condensation in random potentials}}, PhD
  thesis, Fern\-Universität in Hagen, 2019.

\bibitem[PO56]{PO56}
O.~Penrose and L.~Onsager, \emph{{{B}ose--{E}instein condensation and liquid
  helium}}, Phys. Rev. \textbf{104} (1956), 576--584.

\bibitem[Rue99]{ruelle1999statistical}
David Ruelle, \emph{{Statistical mechanics: Rigorous results}}, World
  Scientific, 1999.

\bibitem[Sme86]{PS86}
{P.~de} Smedt, \emph{{The effect of repulsive Interactions on
  {B}ose--{E}instein condensation}}, J. Stat. Phys. \textbf{45} (1986),
  201--213.

\bibitem[SYZ12]{SeiYngZag12}
R.~Seiringer, J.~Yngvason, and V.~A. Zagrebnov, \emph{{Disordered
  {B}ose--{E}instein condensates with interaction in one dimension}}, J. Stat.
  Mech.: Theory and Experiment \textbf{2012} (2012), P11007.

\bibitem[Ven13]{Veniaminov}
N.~A. Veniaminov, \emph{{The existence of the thermodynamic limit for the
  system of interacting quantum particles in random media}},
  Ann.~Henri~Poincar\'{e} \textbf{14} (2013), 63–94.

\bibitem[Ver11]{VerbeureBook}
A.~F. Verbeure, \emph{{Many-body boson systems: Half a century later}},
  Springer-Verlag, 2011.

\end{thebibliography}
	
	\end{document}